\documentclass[12pt,a4paper]{article}
\usepackage{amssymb,amsmath,amsthm}
\usepackage{epsfig}
\usepackage{tikz}
\usepackage{setspace}
\newtheorem{definition}{Definition}[section]
\newtheorem{proposition}[definition]{Proposition}

\newtheorem{theorem}[definition]{Theorem}

\def\checkmark{\tikz\fill[scale=0.4](0,.35) -- (.25,0) -- (1,.7) -- (.25,.15) -- cycle;} 

\usepackage{framed}

\newtheorem{example}[definition]{Example}
\newtheorem{remark}[definition]{Remark}

\newcommand{\RR}{\mathbb{R}}
\newcommand{\BB}{\mathbb{B}}
\DeclareMathOperator*{\argmax}{arg\,max}
\DeclareMathOperator*{\argmin}{arg\,min}

\newcommand{\Img}{\operatorname{Im}}
\newcommand{\fix}{\operatorname{fix}}
\newcommand{\nonfix}{\operatorname{non-fix}}
\newcommand{\maj}{\operatorname{maj}}
\newcommand{\imax}[1]{{#1}\textup{-}\!\max}
\newcommand{\iargmax}[1]{{#1}\textup{-}\!\argmax}
\newcommand{\power}[1]{{\mathcal P}(#1)}
\newcommand{\ignore}[1]{}

\title{A Higher-order Framework for \\ Decision Problems and Games\footnote{We thank seminar participants at the University of Mannheim, the Dagstuhl Perspectives Workshop ``Categorical Methods at the Crossroads'', the Computing in Economics and Finance Conference 2014 in Oslo, and the ``Cogrow'' Workshop in Nijmegen 2014 for helpful comments.}}

\author{Jules Hedges, Paulo Oliva\\School of Electronic Engineering and Computer Science\\Queen Mary University London\vspace{0.5cm}\\
Evguenia Winschel, Viktor Winschel, Philipp Zahn\\Department of Economics\\University of Mannheim}
\date{\today}
\begin{document}
\maketitle
\begin{abstract}
We introduce a new unified framework for modelling both decision problems and finite games based on \emph{quantifiers} and \emph{selection functions}. We show that the canonical utility maximisation is one special case of a quantifier and that our more abstract framework provides several additional degrees of freedom in modelling. In particular, incomplete preferences, non-maximising heuristics, and context-dependent motives can be taken into account when describing an agent's goal. We introduce a suitable generalisation of Nash equilibrium for games in terms of quantifiers and selection functions. Moreover, we introduce a refinement of Nash that captures context-dependency of goals. Modelling in our framework is compositional as the parts of the game are modular and can be easily exchanged. We provide an extended example where we illustrate concepts and highlight the benefits of our alternative modelling approach.

\end{abstract}

\onehalfspacing
\section{Introduction}

In this paper we introduce a new formal framework to reason about decision problems and finite games. This framework has been developed in computer science as a game theoretical approach to proof theory\footnote{Proof theory is a branch of mathematical logic which investigates the structure and meaning of formal mathematical proofs. It has been recently discovered that certain proofs of high logical complexity can be interpreted as computer programs which compute equilibria of suitable generalised games.} and is based on higher order functions \cite{escardo10a,escardo_sequential_2011}. We adopt and develop it further to make it accessible to economics. 

The core concept is the modelling of the agent's goal as a \emph{quantifier}, i.e. a higher-order function of type $(X \to R)\to R$, where $X$ is the set of choices and $R$ is the set of possible outcomes. A corresponding notion is that of a \emph{selection function}, i.e. a higher-order function of type $(X \to R) \to X$ which calculates a choice that ``meets" the desired goal. In Section \ref{sec:quantifiers} we have a brief introduction to higher-order functions, and give the precise definitions of quantifiers and selection functions. We will also illustrate how these provide a powerful and flexible way to model agents. 


In economic theory, agents are typically assumed to maximise an utility function. This is a special case of a quantifier and a particularly structured outcome space, $\max \colon (X \to \RR) \to \RR$. The corresponding selection function is $\argmax \colon (X \to \RR) \to X$.  Our aim with this paper is to show that representing decision problems and games in the more abstract form of quantifiers and selection functions offers three advantages over the special cases of $\max$ and $\argmax$.

%
%


First, selection functions and quantifiers are more \emph{expressive} as they provide additional degrees of freedom in modelling behaviour. Different quantifiers other than the max operator are possible, for instance, decision heuristics where agents do not fully optimise. The outcome space can have any structure and is not restricted to the structure imposed by utility functions (or, equivalently, rational preferences). In particular, one can directly model preference relations that are incomplete. Also, since quantifiers and selection functions take functions as input, context-dependent goals (where not only outcomes matter but also how outcomes come about) can be seamlessly modelled. In Section \ref{sec:decisions} we explain how these generalisations can be done step by step and give various examples. 

As a second advantage, quantifiers and selection functions model not only decisions but also interactions and thereby \emph{unify decision problems and games} in one formal framework. In Section \ref{sec:games} we formally introduce games and a generalised notion of Nash equilibrium in terms of quantifiers. As it is possible to define an equilibrium in terms of selection functions, we also introduce what we call \emph{selection Nash} equilibrium. We prove that Nash equilibrium and selection equilibrium are isomorphic in the case of classical operators max and argmax. Moreover, we prove that, generally, this isomorphism does not hold true: For other quantifiers and selection functions the two different concepts yield different sets of equilibria. In fact, the selection equilibrium is a refinement of the generalised Nash equilibrium. 

The third advantage of our approach is \emph{compositionality} in modelling. We think of a game as consisting of a \emph{global} outcome function determining outcomes for each given play of the game, and \emph{local} quantifiers or selection functions describing each of the particular players' intentions. Selection functions and quantifiers are modular as, for instance, if in a particular game, one would like to consider the consequences of changing the preferable outcomes for one player, only this player's quantifier or selection function need to change. This is of particular advantage in the case of context-dependent goals. With utility functions it may be necessary to redefine the outcome space or to manipulate payoffs by hand such that a context-dependency is encoded. In Section \ref{sec:compiling} we discuss at length the relationship between games based on classical payoff functions and based on selection function as well as how to automatically compile a selection function game into a classical one. 

All these concepts are illustrated in Section \ref{sec:examples} where we provide an extended example. We consider several variants of a beauty contest. While the rules of the contest are fixed, the goals of agents change from variant to variant to illustrate the compositionality of the selection function approach. We begin with considering agents who care about the outcome of the beauty contest. Next, we introduce some agents who only care about choosing the winner of the contest, i.e. the votes of the other players. We show how this concern can be nicely modelled as a context-dependent fixpoint quantifier: The agent wants his vote to be the same as the majority decision of the overall game. Beyond this particular example, we discuss how fixpoint operations in general capture coordination goals of agents at a higher level, highlighting again the expressivity of selection functions. As a last point, this example also teaches the intuition about the selection equilibrium, and how it refines the generalised Nash, as it takes into account goals that depend on the context under which they are attained. 


A notable feature of our approach, which we do not explore in detail in this paper but still consider to be important, is that it is directly implementable in functional programming languages such as Haskell. Hence, large and complex games can be programmed and computer assistance in analysing these games is readily available. This is discussed in Section \ref{sec:conclusion}.

\ignore{
The paper is organised as follows. In the first part of this paper, we introduce the theory of \emph{quantifiers} and \emph{selection functions}  as the foundation of our framework. In the second part, we represent decision problems in terms of quantifiers and selection functions followed by  simultaneous move games in the third part of the paper.
In the fourth part of the paper, we illustrate our approach by discussing an extended example based on Keynesian beauty contest . We exemplify that our approach allows to model goals in a very expressive way. For instance, an agent who wants to vote for the winner of a contest but who is not intrinsically interested in the contestants per se can be modelled as a fixpoint agent, an agent whose sole goal it is to vote such that his vote is in line with the majority. More generally, fixpoint agents capture a high-level representation for the goal of coordination in certain economic situations - such as the beauty contest. 
Moreover, we also illustrate the framework's flexibility by looking at variants of the beauty contest where we combine completely different kinds of agents. All this can be done within the same framework.
In our examples, we also explore the concept of the selection equilibrium which perfectly captures the intuitively plausible equilibria in games 
where agents take into account the actions of other agents in a context-dependent way. 
}

\section{Quantifiers and Selection Functions}
\label{sec:quantifiers}

A \emph{higher order function} (or \emph{functional}) is a function whose domain is itself a set of functions. Given sets $X$ and $Y$ we denote by $X \to Y$ the set of all functions with domain $X$ and codomain $Y$. A higher order function is therefore a function $f : (X \to Y) \to Z$ where $X$, $Y$ and $Z$ are sets.

A simple example of a higher order function is the function that evalulates its argument at a constant point. To give a specific example, we take the sets $\mathbb R$ (real numbers) and $\mathbb Z$ (integers), and pick a constant real number, such as $\pi$. We can then define a function $\Phi : (\mathbb R \to \mathbb Z) \to \mathbb Z$ by the equation $\Phi (f) = f (\pi)$. We can illustrate the behaviour of $\Phi$ by giving it a specific function $f : \mathbb R \to \mathbb Z$ as an input. For example, let $f$ be the function that takes a real number to its integer lower bound. The integer lower bound to $\pi$ is 3, therefore $\Phi (f) = 3$. 


There are familiar examples of higher-order functions, for example in the maximum operator the name of the variable which ranges over the set $X$ is not relevant, i.e. 
\[ \max_{x \in X} f(x) = \max_{y \in X} f(y) \]
and similarly with the variable over which one performs integration, i.e.
\[ \int f(x)\,\mathrm{d}x = \int f(y)\,\mathrm{d}y \]
In this section we define two particular classes of higher-order functions: \emph{quantifiers} and \emph{selection functions}. Subsequently (Section \ref{sec:decisions}) we represent  the classical approach to decision theory via preference relations and $\max$ and $\argmax$ operators within the new formalism. We will also explain how to arrive at the notion of a quantifier via a series of generalisations from utility functions.

\subsection {Quantifiers}
\label{sec:quantifier_preference_util}

Suppose we have an agent $\mathcal A$. We can place $\mathcal A$ into any economic \emph{situation} or \emph{context} and observe his motivations and his choices. We know that he is \emph{deterministic} (or \emph{predictable}) in the sense that his moves are not dependent on chance.\footnote{This is without loss of generality, because we can always allow the set of outcomes to be a set of probability distributions.} 

By a `situation' or `context' we mean an object that encodes all of the relevant information the agent could consider when choosing a move or strategy. Assume our agent is choosing a move in the set $X$, and the set of possible final outcomes is $R$. The context will normally include other agents and all the other choices that together with the choice of our agent $\mathcal A$ will determine a final outcome. If all we care about is the final outcome, then our context can be modelled simply by a function $p \colon X \to R$ that maps each of the agent's moves to a specific outcome. In other words, to give the context of an agent is the same as to define precisely what final outcomes will result after each of the agent's choices. That is all that our agent needs to know about this ``context" in order to make the good choice. 

\begin{definition}[Agent context] For an agent $\mathcal A$ choosing a move from a set $X$, having in sight a final outcome in a set $R$, we call any function $p : X \to R$ a possible \emph{context} for the agent $\mathcal A$. 
\end{definition}

Suppose that $\mathcal A$ makes a decision in the context $p$. Then the agent will consider some outcomes to be \emph{good} (or \emph{acceptable}), and other outcomes to be bad. We are going to allow the set of outcomes that the agent considers good to be totally arbitrary. 

\begin{definition}[Quantifier, \cite{escardo10a,escardo_sequential_2011}] Mappings $$\varphi : (X \to R) \to \power{R}$$ from contexts $p : X \to R$ to sets of outcomes $\phi(p) \subseteq R$ are called \emph{quantifiers}\footnote{The terminology comes from the observation that the usual existential $\exists$ and universal $\forall$ quantifiers of logic can be seen as operations of type $(X \to \BB) \to \BB$, where $\BB$ is the type of booleans. Mostowski \cite{Mostowski(1957)} also called arbitrary functionals of type $(X \to \BB) \to \BB$ \emph{generalised quantifiers}. We are choosing to generalise this further by replacing the booleans $\BB$ with an arbitrary type $R$, and allowing for the operation to be multi-valued.}. If $\varphi(p) \neq \varnothing$ for all $p \colon X \to R$ we say that the quantifier $\varphi$ is \emph{total}.
\end{definition}

We be will modelling agents $\mathcal A$ as quantifiers $\varphi_{\mathcal A}$, and in such cases we wish to think of $\varphi(p)$ as the set of outcomes the agent $\mathcal A$ considers preferable in each context $p \colon X \to R$. Our main objective in this paper is to convince the reader that this is a general, modular, and highly flexible way of describing an agent's goal or objective.

The classical example of a quantifier is utility maximisation. In this case the set of outcomes is $R = \mathbb R^n$, where the $i$th element represents the utility of the $i$-th player. Given a context $p : X \to \mathbb R^n$, the good outcomes for the $i$th player are precisely those for which the $i$th coordinate, i.e. his utility function, is maximal. This quantifier is given by
\[ \imax{i}(p)= \{ r \in \Img(p)\mid r_i \geq p(x')_i \text{ for all } x' \in X \} \]
where $\Img(p)$ denotes the image of the function $p \colon X \to R$. Note that the image mapping $\Img(\cdot)$ is itself a quantifier, although not a particularly interesting one, as it corresponds to the agent that considers any possible attainable outcome to be good.

\begin{definition}[Context-independent quantifiers] A quantifier $\varphi \colon (X \to R) \to \power{R}$ is said to be \emph{context-independent} if the value $\varphi(p)$ only dependents on $\Img(p)$, i.e.
\[ \Img(p) = \Img(p') \implies \varphi(p) = \varphi(p'). \]
Hence, a quantifier $\varphi$ will be called \emph{context-dependent} if for some contexts $p$ and $p'$, with $\Img(p) = \Img(p')$, the sets of preferred outcomes $\varphi(p)$ and  $\varphi(p')$ are different.
\end{definition} 

Intuitively, a context-dependent quantifier will select good outcomes not just based on which outcomes are possible, but will also take into account how the outcomes are actually achieved. It is easy to see that the quantifier $\imax{i}(p)$ is context-independent, since it can be written as a function of $\Img (p)$ only. An example of a context-dependent quantifier is the fixpoint operator. Recall that a fixpoint of a function $f : X \to X$ is a point $x \in X$ satisfying $f(x) = x$. When the set of moves is equal to the set of outcomes there is a quantifier whose good outcomes are precisely the fixpoints of the context. If the context has no fixpoint and the agent will be equally happy (or equally unhappy) with any outcome, then the quantifier is given by
\[ \fix : (X \to X) \to \mathcal P (X) \]
\[ \fix (p)= \begin{cases}
\{ x \in X \mid p(x) = x \} &\text{ if nonempty } \\
X &\text{ otherwise}
\end{cases} \]
Clearly $\fix(\cdot)$ is context-dependent, since we could have two maps $p, p' \colon X \to X$ having the same image set $\Img(p) = \Img(p')$ but with $p$ and $p'$ having different sets of fixed points. 

We will discuss several examples of fixpoint quantifiers at great length in Section \ref{sec:examples}.
\subsection{Selection Functions}

Just as a quantifier tells us which outcomes an agent considers good in each given context, one can also consider the higher-order function that determines which \emph{moves} an agent considers good in any given context. 

\begin{definition}[Selection functions] A \emph{selection function} is any function of the form
\[ \varepsilon : (X \to R) \to \mathcal P (X) \]
If $\varepsilon(p) \neq \varnothing$ for all $p \colon X \to R$ we say that the selection function is \emph{total}. 
\end{definition}

We will mainly consider total selection functions and quantifiers, because the agent must always have some preferred outcomes and moves.

In the computer science literature where selection functions have been considered previously \cite{escardo10a,escardo_sequential_2011} the focus was on single-valued ones. However, as multi-valued selection functions are extremely important in our examples we have adapted the definitions accordingly.

Similarly to quantifiers, the canonical example of a selection function is maximising one coordinate in $\mathbb R^n$, defined by
\[ \iargmax{i}(p)= \{ x \in X \mid p(x)_i \geq p(x')_i \text{ for all } x' \in X \} \]
Even in one-dimensional $\mathbb R^1$ the $\argmax$ selection function is naturally multi-valued: a function may attain its maximum value at several different points.

There is an important relation between quantifiers and selection functions called attainment. Intuitively this means that the outcome of a good move should be a good outcome. 

\begin{definition} Given a quantifier $\varphi : (X \to R) \to \mathcal P (R)$ and a total selection function $\varepsilon : (X \to R) \to \mathcal P (X)$, we say that $\varepsilon$ \emph{attains} $\varphi$ iff for all contexts $p : X \to R$ it is the case that
\[ x \in \varepsilon (p) \implies p (x) \in \varphi (p)\]
Clearly, if $\varphi$ is attainable then it is also total.
\end{definition}

The attainability relation holds between the quantifier $\max_i$ and the selection function $\argmax_i$. The fixpoint quantifier is also itself a selection function, and it attains itself since
\[ x \in \fix (p)\implies p (x) \in \fix(p)\]

Given a selection function $\varepsilon$, we can form the `smallest' quantifier which it attains as follows.

\begin{definition} Given a selection function $\varepsilon \colon (X \to R) \to \power{X}$, define the quantifier $\overline{\varepsilon} \colon (X \to R) \to \power{R}$ as
\[ \overline{\varepsilon} (p)= \{ p (x) \mid x \in \varepsilon (p)\}. \]
\end{definition}

We can use the mapping between selection functions and quantifiers in order to transfer back properties of quantifiers to selection functions. For instance, we can then say that a selection function $\varepsilon \colon (X \to R) \to \power{X}$ is context-independent if its associated quantifier $\overline{\varepsilon} \colon (X \to R) \to \power{R}$ is context-independent. Here are the three main properties of the map $\varepsilon \mapsto \overline{\varepsilon}$.

\begin{proposition} Given any selection function $\varepsilon \colon (X \to R) \to \power{X}$ the following three properties are easy to check:
\begin{itemize}
    \item[(i)] {\bf Totality}. If $\varepsilon$ is total then so is the quantifier $\overline{\varepsilon}$.
    \item[(ii)] {\bf Attainability}. If $\varepsilon$ is total then $\varepsilon$ attains $\overline{\varepsilon}$.
    \item[(iii)] {\bf Minimality}. $\overline{\varepsilon}(p) \subseteq \varphi(p)$, for any quantifier $\varphi$ attainable by $\varepsilon$.
\end{itemize}
\end{proposition}

Therefore, the good outcomes according to the quantifier $\overline{\varepsilon} \colon (X \to R) \to \power{R}$ are exactly the outcomes resulting from good moves according to the selection function $\varepsilon \colon (X \to R) \to \power{X}$.

Conversely, for a quantifier $\varphi \colon (X \to R) \to \power{R}$ we can define a corresponding selection function as follows. 

\begin{definition} Given a quantifier $\varphi \colon (X \to R) \to \power{R}$, define the selection function $\overline{\varphi} \colon (X \to R) \to \power{X}$ as
\[ \overline{\varphi}(p) = \{ x \mid p(x) \in \varphi(p) \}. \]
\end{definition}

We use the same ``overline" notation, as it will be clear from the context whether we are applying it to a quantifier or a selection function.

\begin{proposition} Given any quantifier $\varphi \colon (X \to R) \to \power{R}$ the following three properties are easy to check:
\begin{itemize}
    \item[(i)] {\bf Totality}. If $\varphi$ is attainable then the selection function $\overline{\varphi}$ is total.
    \item[(ii)] {\bf Attainability}. If $\varphi$ is attainable then $\overline{\varphi}$ attains $\varphi$.
    \item[(iii)] {\bf Maximality}. $\varepsilon(p) \subseteq \overline{\varphi}(p)$, for any selection function $\varepsilon$ attaining $\varphi$.
\end{itemize}
\end{proposition}

Let us briefly reflect on the game-theoretic meaning of attainability, and the translations between quantifiers and selection functions. Suppose we have a quantifier $\varphi$ which describes the outcomes that an agent considers to be good. The quantifier might be \emph{unrealistic} in the sense that it has no attainable good outcome. For example, an agent may consider it a good outcome if he received a million dollars, but in his current context there may just not be a move available which will lead to this outcome. Given a context $p$, the set of attainable outcomes is precisely the image of $p$. A \emph{realistic} quantifier is simply a quantifier in which every context with a good outcome has an attainable good outcome. We can write it in symbols as
\[ \varphi (p) \neq \varnothing \implies \varphi (p) \cap \Img (p) \neq \varnothing. \]

\begin{proposition} For a total quantifier $\varphi$ the following are equivalent
    \begin{itemize}
    \item $\varphi$ is realistic,
    \item $\varphi$ is attainable,
    \item $\overline{\varphi}$ is total.
    \end{itemize}
\end{proposition}

Thus total selection functions are a way to describe total realistic quantifiers. One can consider translating quantifiers into selection function and back into quantifiers, or conversely.

\begin{proposition} For all $p \colon X \to R$ we have $\overline{\overline{\varphi}}(p) = \varphi(p)$ and $\varepsilon(p) \subseteq \overline{\overline{\varepsilon}}(p)$.
\end{proposition}

The proposition above shows that on quantifiers the double-overline operation calculates the same quantifier we started with. However, on selection functions the the mapping $\varepsilon \mapsto \overline{\overline{\varepsilon}}$ can be viewed as a \emph{closure} operator. Intuitively, the new selection function $\overline{\overline{\varepsilon}}$ will have the same good \emph{outcomes} as the original one, but it might consider many more \emph{moves} to be good as well, as it does not distinguish moves which both lead to equally good outcomes. As such, one can think of multi-valued selection function as a finer way to describe agents' motivations. As we will see in Section \ref{sec:games}, selection functions are also crucial to define a useful equilibrium concept for  games with agents who have context-dependent quantifiers, which is finer than previously considered equilibrium concepts.

\begin{remark} The theory of quantifiers and selection functions has been developed in stages. Single-valued selection functions and quantifiers in the general form used here first appeared in \cite{escardo10a}, unifying earlier definitions in proof theory and type theory. That is also where the connection between selection functions and game theory also first established. Multi-valued quantifiers appeared in \cite{escardo_sequential_2011}, which allows us to capture more important examples in a more natural way. The connections between selection functions and game theory were explored in more depth in \cite{escardo12} and \cite{hedges13}, and the latter contains the definition of attainment given here. Finally \cite{hedges14a} contains the terminology \emph{context} and the definition of a realistic quantifier.
\end{remark}

\section{Decisions}
\label{sec:decisions}

In this section we relate the concepts of quantifiers and selection functions to the standard utility approach in decision theory. We show that the choices motivated by utility maximisation are a special case of context-independent quantifiers. Next, we show that the set of context-independent quantifiers contains elements that allow to model decisions with fewer assumptions than utility theory and also in a more expressive way. Lastly, we consider quantifiers that are context-dependent. We provide examples along the way to illustrate the concepts.

Suppose $R$ is the set of possible final outcomes, and each agent $i$ has a partial order relation $\succeq_i$ on $R$, so that $x \succeq_i y$ means that agent $i$ prefers the outcome $x$ to $y$. These partial orders lead to choice functions $f_i : \mathcal P (R) \to \mathcal P (R)$ where $f_i(S)$ are the maximal elements in the set of possible outcomes $S$ with respect to the order $\succeq_i$. Note that these $f_i$ satisfy $f_i (S) \subseteq S$, and $f_i (S) \neq \varnothing$ for non-empty $S$.

Every such $f_i$ can be turned into a quantifier $\varphi_i$ in a generic way, using the fact that the image operator is a higher-order function $\Img : (X \to R) \to \mathcal P (R)$:
\[ (X \to R) \xrightarrow{\Img} \mathcal P (R) \xrightarrow{f_i} \mathcal P (R) \]
so that $f_i \circ \Img \colon (X \to R) \to \mathcal P (R)$ are quantifiers. 

\begin{proposition} Assume $|X| \geq |R|$. Then a quantifier $\phi \colon (X \to R) \to \power{R}$ is context-independent if and only if $\phi = f \circ \Img$, for some choice function $f \colon \power{R} \to \power{R}$.
\end{proposition}
\begin{proof} If $\phi = f \circ \Img$ then clearly $\phi$ is context-independent. For the other direction, note that since $|X| \geq |R|$ we have for any subset $S \subseteq R$ a map $u_S \colon X \to R$ such that $\Img(u_S) = S$. Hence, assume $\phi$ is context-independent and define $f(S) = \phi(u_S)$. Clearly,
\[ \phi(p) = \phi(u_{\Img(p)}) = f(\Img(p)) \]
where the first step uses that $\phi$ is context-independent and that $\Img(p) = \Img(u_{\Img(p)})$ by the assumption on the family of maps $u_S$.
\end{proof}

Agents who are defined by context-independent quantifiers are choosing the set of good outcomes simply by ranking the set of outcomes that can be achieved in a given context; but are forgetting all the information about how each of the outcomes arise from particular choices of moves. For instance, we might have a set of actions that will lead us to earn some large sums of money. Some of these, however, might be illicit. A maximising agent defined in a context-independent way would choose the outcome that gives himself the maximum return. If we have control over which actions lead to which outcomes, we might consider other choices as preferable.

\begin{proposition} Whenever $f_i$ is a choice function arising from a partial order $\succeq_i$, then the context-independent quantifier $\phi_i = f_i \circ \Img$ is attainable (and hence realistic).
\end{proposition}
\begin{proof} Define a selection function for $\phi_i$ as
\[ \varepsilon_i(p) = \{ x \mid \mbox{$p(x)$ is $\succeq_i$-maximal in $\Img(p)$} \}. \]
Clearly if $x \in \varepsilon_i(p)$ then $p(x) \in f_i(\Img(p)) = \phi_i(p)$. Also, $\varepsilon(p)$ is total, since $\Img(p)$ is always non-empty.
\end{proof}


\subsection{Rational Preferences and Utility Functions}

The usual approach to model behaviour in economics is to either postulate a preference relation on the set of alternatives or to directly assume a utility function \cite{kreps2012microeconomic}. Typically, a certain structure is imposed on preference relations mainly for two reasons: either because additional structure is deemed to be a characteristic of an agent's rationality\footnote{This issue has been intensely debated, see \cite{Richter_rational_1971,mas-colell_microeconomic_1995,kreps2012microeconomic}.}, or because one wishes to work with utility functions. It is a classical result that for utility functions to exist, preferences relations have to be \emph{rational} \cite{kreps2012microeconomic}.

Now, rational preferences and utility functions are special cases of the generic construction of a context-independent quantifier we outlined in the last section. They are special because (i) we impose additional structure on $R$, that is, $\succeq_i$ is a total preorder and (ii) we focus on one particular $f_i$, that is, $f_i : \mathcal P (R) \to \mathcal P (R)$ defined by
\[ f_i (S) = \{ \succeq_i \text{-maximal elements of } S \} \] 

A rational preference relation can always be represented by a utility function. Translated into the selection function approach, the utility function can be characterised as the environment which is a mapping $p \colon X \to \RR$, attaching a real number to each element of the set of choices $X$. So, we can define the quantifier
\[
  \phi (p) = \max p
\]
which is attained by the selection function
\[
  \varepsilon (p) = \arg\max p
\]
Note the types $\phi \colon (X \rightarrow \RR) \rightarrow \mathcal P (\RR)$ and $\varepsilon \colon (X \rightarrow \RR) \rightarrow X$ respectively and that $\overline{\varepsilon} (p) = \phi (p)$. Thus, $\max$ and $\arg \max$ operators, which are universally used in the economic literature, become the prototypical examples of a context-independent quantifier and a selection function attaining it.

We can use selection functions in situations where instead of using utility functions we directly work with preferences. We demonstrate this with the following example.

\begin{example} \label{examp:decision1} Consider a simple decision problem of an agent who has to choose between three alternatives $X=\{A, B, C\}$.  Assume the agent has the following (total and transitive) preference order\footnote{We use an index in the ordering as we will consider further refinements of this example in the rest of this section.} $\{ A \succeq_1 B \;,\; B \succeq_1 C \;,\; A \succeq_1 C \}$. 
To use the standard maximisation tools one would associate each alternative with some numerical payoff or define an utility function for the agent. For example, assume that the payoff of the agent is equal to $p_A = 1$ if he chooses alternative $A$, $p_B = 0.5$ if he chooses $B$ and $p_C = 0$ otherwise. Then the utility function $u(x) = p_x$ represents the preferences of the agent. \\[1mm]
We can work with this utility function in the selection function framework. The utility function is then just the  environment $p \colon X \to \RR$ and the quantifier is the max operator, the selection function is the argmax operator, respectively. \\[1mm]
But we also can work directly with the preference order. Then, our environment function is $p \colon X \to X$. And we obtain a quantifier and a selection function, which takes the maximal element with respect to the preference ordering $\succeq_1$, namely
\[
\phi_1(p) = \max_{x\in (X, \succeq_1)} p(x) 
\quad \mbox{and} \quad 
\varepsilon_1 (p) = \argmax_{x\in (X, \succeq_1)} p(x).
\]
\end{example}

\subsection{Beyond Rational Preferences }

The generic construction of context-independent quantifiers instantiates choices based on rational preferences (or equivalently on utility maximisation) as special cases. In this section we show that we can go beyond these cases by allowing for a different structure on $R$ or by allowing for a different $f_i$ (or by relaxing both).

Utility functions are considered as a very convenient tool to represent and analyse choice behaviour. Still, the assumption that the preorder is total, which guarantees the existence of a utility function, is demanding and in fact more demanding than is necessary to rationalise choice behaviour \cite{Richter_rational_1971}. Secondly, when taking the perspective of preferences, from a positive as well as a normative viewpoint, there are good reason why a rational decision-maker may exhibit ``indecisiveness", meaning that his preference for a pair of outcomes is not defined \cite{Aumann1962}. 
Thirdly, consider a situation where the economist or some other agents/principal has only partial information about the preferences of an agent and considers him ``as if'' he has incomplete preferences \cite{Dubra_et_al2004_comple}.
Lastly, $R$ may be a set of alternatives to be chosen by a group of agents. Even if each individual's preferences are complete, the aggregate social welfare ordering does not have to be \cite{Ok2002}.

There have been various attempts to change standard formalisms to allow for an utility theory without the need to fulfil the completeness assumption.\footnote{For an important early contribution see \cite{Aumann1962}. More recent contributions include \cite{Ok2002} for utility representations in certain environments and \cite{Dubra_et_al2004_comple} for uncertain environments. See also references in \cite{Ok2002}.} When working with quantifiers and selection functions, the set of outcomes $R$ can have \textit{any} order. In particular, the preference relation does not have to be total. That is, given any preference relation $\succeq \, \subseteq R \times R$, an agent chooses the best alternatives as outlined above. So, one can very easily consider choices not in the scope of utility functions without the need to change the framework. To be clear, the selection function corresponding to the preference ordering is
\[ \argmax (p) = \{ x \in X \mid r \succeq p(x) \implies r \not\in \operatorname{Im} (p) \} \]
ie. a maximal outcome is one which is not known to be worse than any attainable outcome.

\begin{example} Continuing from our earlier Example \ref{examp:decision1}, suppose the agent, who has to choose between alternatives $X=\{A, B, C\}$, prefers $A$ over $B$, but he has incomplete preferences and cannot rank the alternative $C$. In this case the order relation is simply $A \succeq_2 B$. In our setting this is seamlessly dealt with, as we are simply picking the maximal element with respect to this partial ordering. In particular, the quantifiers and selection functions are the same, only the ordering has changed
\[
\phi_2(p) = \max_{x\in (X, \succeq_2)} p(x) 
\quad \mbox{and} \quad
\varepsilon_2 (p) = \argmax_{x\in (X, \succeq_2)} p(x). 
\]
\end{example}

Note that it is not directly possible to use the max operator and a utility function for the last example. However, using selection functions and working directly on the preorders, we can just use the same operator as for preference orders that are total.

\subsection{Beyond Maximisation and Standard Rationality}
The utility approach is intimately linked to the assumption that the agent fully optimises. The behavioural economic literature as well as the psychological literature have documented deviations from optimising behaviour, and have collected various decision ``heuristics'' \cite{camerer2011behavioral,kahneman2011thinking}. Quantifiers provide a nice way to model such deviations. Moreover, even situations that can be modelled with utility functions may have (more) natural representations in the quantifier framework.

\begin{example}

Consider a simple heuristic of a person ordering wine in a restaurant. Suppose he always chooses the second most-expensive wine. In terms of selection functions, let $X$ be the set of wines available in a restaurant, and $p: X \rightarrow \RR$ the price attached to each wine $x_i$ ($i=1, ...,N$) on the menu. Denote with $r_i = p(x_i)$ the price of wine $x_i$. Given a maximal strict chain $r_n > r_{n-1} > \ldots > r_1$ in $\RR$, let us call $r_{n-1}$ a sub-maximal element. The ``goal" of the agent can be described by the quantifier
\[ \phi_{>}(p^{X \to \RR}) = \{ \mbox{sub-maximal elements with respect to $>$ within $\Img(p)$} \}. \]
Such quantifiers are attainable with selection functions
\[ \varepsilon_{>}(p^{X \to \RR}) = \{ x \mid \mbox{$p(x)$ is a sub-maximal element of $\Img(p)$} \} \]
since clearly $p(x) \in \phi_>(p)$ for $x \in \varepsilon_>(p)$.

\end{example}

A crucial point of the above example is the additional degree of freedom of modelling as it is possible to vary the choice operator itself and not being automatically restricted to the max operator and to consider behaviour to be rationalised by preferences.\footnote{In some sense, our viewpoint is similar to choice rules or choice functions. The behaviour is the focal point of analysis. See \cite{mas-colell_microeconomic_1995} or \cite{rubinstein2006lecture}.} 

Obviously, one could rationalise the above choice as the outcome of a maximisation. One could redefine preferences and utility functions such that the outcome of the maximisation is just the second most expensive wine.\footnote{Note, if the prices of the wines represented preferences, a rationalisation of 2nd best choices were not possible (see \cite{rubinstein2006lecture}).}  However, while equivalent in outcome, the causal model of behaviour is different. The classical approach would force the choice to be rational, whereas in our setting this question remains open. The quantifier just formally describes an agent's behaviour. It could be that the choice pattern is a habitual heuristic or it could be the reduced form pattern of rational decision-making in a larger context.

Of course, instead of using the second most expensive wine, one could consider alternative heuristics, such as choosing the wine closest to the average price of all available wines on the menu, or within a class of wines, etc..

Moreover, one could also combine this heuristic with preferences. Say, the guest is a fan of white wines, and he strictly prefers Chardonnay over Riesling. One could model the agent as first restricting the choices to the wines that are Chardonnay (if available) and then apply his second most expensive decision heuristic to the class of Chardonnay available.


\subsection{Context-Dependent Decision Problem}\label{sec:context_dec}

So far, we have focused only on the generic context-independent quantifiers. As the last examples illustrate with this construction we can already go beyond choices motivated by rational preferences. Yet, we can do more. We can allow for quantifiers that do not only take the image of $p$ as input but the complete function.

Next, we provide an example to illustrate that this opens up a complete new dimension. Indeed, with context-dependent quantifiers it is possible to go far beyond what can be modelled using utility functions.

\begin{example} [Keynesian beauty contest] \label{keynesian-ex1} Consider the following situation: there are three players, the judges $J=\{J_1, J_2, J_3\}$. Each judge votes for two contestants $A$ and $B$. 
The set of outcomes is given by $X=\{A, B\}$ denoting the winner of the contest. 
The winner is determined by the simple majority rule of type $\operatorname{maj}: X\times X\times X \rightarrow X$.  To begin with, we assume that the judges rank the contestants according to a preference ordering. For example, suppose judges 1 and 2 prefer $A$ and judge 3 prefers $B$. We consider the decision problem of the first judge who observes the choices of two others.  Suppose, judge 1 observes that judge 2 votes for $A$ and judge 3 votes for $B$, and he has to decide how to cast his vote. The order relation of the first judge is $A \succeq_1 B$. Thus we obtain a quantifier and a selection function, which is taking the maximal element with respect to the ordering, as before: 
\[
\phi_1(p) = \max_{x \in (X, \succeq_1)} p(x)
\quad\quad
\varepsilon_1 (p) = \argmax_{x \in (X, \succeq_1)} p(x)
\]
Now, assume judges 2 and 3 continue to prefer contestants $A$ and $B$ respectively, but judge 1 has different preferences: he prefers to support the winner of the contest. He is only interested in voting for the winner of the contest and he has no preferences for the contestants per se.  Again, judge 1 observes that judges 2 and 3 vote for contestants $A$ and $B$. He can be described by a fixed point operator: 
\[
\varphi_1'(p) = \varepsilon_1'(p) = \{ x \in \{A, B\} \mid p(x) = x \}
\]
\end{example}

In practice, most functions do not have a fixed point and so the fixed point quantifier will often give the empty set. For the purposes of modelling a particular situation we might want to `complete' $\varphi$ in different ways, describing what an agent might do in the event that no fixed point exists. Our examples in Section \ref{sec:examples} suggest that the best way to do this is to return the entire set $X$, modelling the fact that the agent is equally `happy' (or unhappy) with any choice if no good choice exists.

The situation above becomes far more interesting when considered as a game in which several agents with potentially different concerns cast a vote. We analyse this in detail in Section \ref{sec:examples}.

\section{General Games} 
\label{sec:games}

Quantifiers and selection functions as introduced in the previous section for decision problems can be directly adapted to model games. All the additional expressiveness introduced before can be used in games as well. In this section, we first define a game and the adequate Nash equilibrium concept.  Moreover, we define a refinement of Nash equilibria that becomes relevant with context-dependent quantifiers.

\begin{definition}[General Games] \label{def-game-general} A general $n$-players game, with a set $R$ of outcomes and sets $X_i$ of strategies for the $i$th player, consists of
\begin{enumerate}
\item for each player $1 \leq i \leq n$, a selection function \[ \varepsilon_i : (X_i \to R) \to \mathcal P (X_i) \] representing that player's preferred strategies in each game context.
\item the outcome function
\[ q : \prod_{i =1 }^n X_i \to R \]
i.e., a mapping from the strategy profile to the final outcome.
\end{enumerate}
\end{definition}

Intuitively, we think of the outcome function $q$ as representing the `situation', or the rules of the game, while we think of the selection functions as describing the agents. Thus we can imagine the same agent in different situations, and different agents in the same situation. This allows us to decompose a modelling problem into a global and a local part: modelling the situation and modelling the players.

\begin{remark}[Strategic game \cite{OsbRub94}] \label{normal-form-remark} The ordinary definition of a strategic game of $n$- players with standard payoff functions is a particular case of Definition \ref{def-game-general} when
\begin{itemize}
\item for each player $i$ set of strategies $X_i$
\item the set of outcomes $R$ is $\mathbb R^n$, modelling the vector of payoffs obtained by each player,
\item the selection function of player $i$ is $\iargmax{i} \colon (X_i \to \RR^n) \to {\mathcal P}(X_i)$, i.e. $\argmax$ with respect to the $i$-th coordinate, representing the idea that each player is solely interested in maximising their own payoff,
\item the $i$-th component of the outcome function $q \colon \prod_{i=1}^n X_i \to \RR^n$ can be viewed as the payoff function $q_i \colon \prod_{j =1 }^n X_j \to \RR$ of the $i$-th player.
\end{itemize}
In the following when we refer to \emph{normal forms games} we mean the special case of Definition \ref{def-game-general} with outcome type and selection functions as above. 
\end{remark}

We illustrate this by modelling the classic ``Battle of the Sexes" game using our framework.

\begin{example}[Battle of the Sexes] \label{ex-cbos} A couple agrees to meet together, but both cannot remember if they will be attending the ballet (B) or a football match (F). The husband prefers football over ballet, while the wife prefers ballet over football. But irrespective of their personal preferences, they would rather be together than by themselves in different places. In this game the set of outcomes is $R = \RR^2$ and the possible strategies for both players are $X_h = X_w = \{B,F\}$. What is normally described as the \emph{payoff matrix} for us shall be viewed as the \emph{outcome function} $X_w \times X_h \to R$, i.e. a mapping from strategy profiles to outcomes. For the game in question the outcome function is shown in Table \ref{tab:BattleofSexes-classic}.
\begin{table}[t!]
\caption{Battle of the Sexes}
\begin{center}
\begin{tabular}{|c|c|c|c|l|}\hline
Strategy 	& Ballet    & Football      \\ \hline \hline
Ballet 	    & 3,2       & 1,1    		\\ \hline
Football    & 0,0       & 2,3         	\\ \hline
\end{tabular}
\end{center}
\label{tab:BattleofSexes-classic}
\end{table}
As standard in classical game theory, both players in this case are trying to maximise their corresponding coordinate of the outcome tuple $(r_w, r_h)$, i.e. the wife wants to maximise $r_w$ whereas the husband would like to maximise $r_h$. Therefore, the selection functions for the two players are the maximisation functions
\begin{itemize}
\item $\varepsilon_w (p) = \{ r_w \, | \, p(r_w) = \max_{x \in X_w} \pi_1(p(x)) \}$
\item $\varepsilon_h (p) = \{ r_h \, | \, p(r_h) = \max_{x \in X_h} \pi_2(p(x)) \}$
\end{itemize}
where $\pi_1, \pi_2 \colon \RR^2 \to \RR$ are the first and second projections, respectively. Note the types of the selection functions are $\varepsilon_w \colon (X_w \to R) \to \power{X_w}$ and $\varepsilon_h \colon (X_h \to R) \to \power{X_h}$.
\end{example}

In order to illustrate the need for multi-valued selection functions, let us also consider an extension of the Battle of the Sexes game where an intermediate possibility is included, namely going to the cinema (\emph{C}). Consider also that if the husband had to go alone somewhere he would prefer football over cinema, and cinema over ballet. The wife on the other hand, if she had to go alone, she would prefer ballet over cinema, and cinema over football. If they are together then the husband would prefer to be at the football, but would consider ballet and cinema almost as nice. For the wife, if they are together she would rather be at the ballet, but would consider being at the football or cinema equally pleasant.

\begin{example}[Extended Battle of the Sexes I] \label{ex-bos-I} In this extended version of the Battle of Sexes the set of outcomes is $R = \RR^2$ and the possible strategies for both players are $X_h = X_w = \{B,C,F\}$. The outcome function is shown in Table \ref{tab:BattleofSexes}.
\begin{table}[t!]
\caption{Extended Battle of the Sexes}
\begin{center}
\begin{tabular}{|c|c|c|c|}\hline
Strategy 	& Ballet 	     & Cinema   & Football      \\ \hline \hline
Ballet 	    & 3,2            & 2,1      & 2,2    		\\ \hline
Cinema      & 1,0            & 2,2      & 1,2       	\\ \hline
Football    & 0,0            & 0,1      & 2,3         	\\ \hline
\end{tabular}
\end{center}
\label{tab:BattleofSexes}
\end{table}
But notice that although we have extended the game by expanding the sets of strategies, and hence the outcome function, the selection functions of both husband and wife are still the same as in Example \ref{ex-cbos}, as they both still aim to maximise their own payoffs.
\end{example}

\subsection{Nash Equilibrium}

Consider a strategy profile $\mathbf x \in \prod_{i=1}^n X_i$. The outcome of this strategy profile is $q (\mathbf x)$. We want to define the context in which one player unilaterally changes his strategy. This is given by the function
\[ 
\mathcal U^q_i (\mathbf x) (x) = q (\mathbf x [i \mapsto x])
\]
of type
\[ 
\mathcal U^q_i : \prod_{j = 1}^n X_j \to (X_i \to R).
\]
Here $\mathbf x [i \mapsto x]$ is the tuple obtained from $\mathbf x$ by replacing the $i$th entry with $x$.
We call the $n$ functions $\mathcal U^q_i$ ($1 \leq i \leq n$) the \emph{unilateral maps} of a given strategy profile $\mathbf x$. They were introduced in \cite{hedges13} in which it is shown that the proof of Nash's theorem amounts to showing that the unilateral maps have certain topological (continuity and closure) properties. The concept of a context was introduced later in \cite{hedges14a}, so now we can say that $\mathcal U^q_i (\mathbf x) : X_i \to R$ is the context 
in which the $i$th player has unilaterally changed his strategy, so we call it a unilateral context.

\begin{example}[Extended Battle of the Sexes II] Continuing from Example \ref{ex-bos-I}, let us illustrate the notion of a unilateral map. Consider the wife (w) and the strategy profile $\mathbf x = (B,C)$, i.e. she decided to go the the ballet and he goes to the cinema. Her unilateral context with respect to this strategy profile is $\mathcal U^q_w (B,C)(x) = q(x,C)$. Or, unfolding the definition of $q$ this can be written more explicitly as
\[
\mathcal U^q_w (B,C)(x) = 
\begin{cases}
(2,1) &\text{ if } x = B \\
(2,2) &\text{ if } x = C \\
(0,1) &\text{ if } x = F.
\end{cases}
\]
The map describes what are the possible payoffs as the wife explores her different choices given that the choice of the husband is fixed. Using her quantifier $\overline{\varepsilon_w}$ on this unilateral context $\mathcal U^q_w (B,C)$ we obtain her preferred outcomes for the context 
\[
\overline{\varepsilon_w} (\mathcal U^q_w (B,C)) = 
\{ (r_w, r_h) \, | \, r_w \in \max_{x \in X_w} \pi_1(\mathcal U^q_w (B,C)(x)) \} = \{(2,1),(2,2)\}
\]
since $\max_{x \in X_w} \pi_1(\mathcal U^q_w (B,C)(x)) = \{ 2 \}$. Similarly, for the husband, his unilateral context with respect to this strategy profile $\mathbf x = (B,C)$ is
\[
\mathcal U^q_h (B,C)(x) = 
\begin{cases}
(3,2) &\text{ if } x = B \\
(2,1) &\text{ if } x = C \\
(2,2) &\text{ if } x = F.
\end{cases}
\]
Hence, his preferred outcomes for this context are
\[
\overline{\varepsilon_h} (\mathcal U^q_h (B,C)) = 
\{ (r_w, r_h) \, | \, r_h \in \max_{x \in X_h} \pi_2(\mathcal U^q_h (B,C)(x)) \} = \{(3,2),(2,2)\}.
\]
\end{example}

Using this notion of a unilateral map we can abstract the classical definition of Nash equilibrium to general games defined by selection functions as follows. 

\begin{definition}[General Nash equilibrium] \label{def-gen-nash} Given a general game $(\varepsilon_i, q)$, we say that a strategy profile $\mathbf x$ is in Nash equilibrium iff \[ 
q (\mathbf x) \in \overline{\varepsilon_i} (\mathcal U^q_i (\mathbf x)) 
\]
for all players $1 \leq i \leq n$.
\end{definition}

As with the usual notion of Nash equilibrium, we are also saying that a strategy profile is in Nash equilibrium if no player has a motivation to unilaterally change their strategy. This is expressed formally by saying that preferred outcomes, specified by the selection function when applied to the unilateral context, contain the outcome obtained by sticking with the current strategy. We illustrate now how this notion indeed coincides with the usual notion when selection functions are maximisation functions:

\begin{example}[Extended Battle of the Sexes III] There are obviously four strategy profiles which are Nash equilibria, namely $\mathbf x = (B,B)$ and $\mathbf x = (C,C)$ and $\mathbf x = (F,F)$ and $\mathbf x = (B,F)$, in the standard sense. Let us see how $\mathbf x = (B,C)$ is not an equilibrium, also in our sense. Consider first the wife (w). As we have calculated her set of preferred outcomes in her unilateral context $\mathcal U^q_w (B,C)$ is $\{(2,1),(2,2)\}$, and indeed
\[ 
q (B,C) = (2,1) \in \{(2,1),(2,2)\} = \overline{\varepsilon_w} (\mathcal U^q_w \mathbf (B,C)) 
\]
so that the wife is happy with the current choice (she is happy to be alone, as long as she is at the ballet). On the other hand, for the husband we have calculated that his preferred outcomes in the unilateral context $\mathcal U^q_h (B,C)$ are $\{(3,2),(2,2)\}$, which in this case does not include the current outcome
\[ 
q (B,C) = (2,1) \notin \{(3,2),(2,2)\} = \overline{\varepsilon_h} (\mathcal U^q_h (B,C)).
\]
Hence, the husband can improve his situation by either going to the football on his own or joining his wife at the ballet. Similarly to above one can verify that all four standard equilibria are also equilibria in the sense of Definition \ref{def-gen-nash}.
\end{example}

\subsection{Selection Equilibrium}\label{sec:context_Nash}

The definition of Nash equilibrium is based on quantifiers. However, we can also use selection functions directly to define an equilibrium condition.

\begin{definition}[Selection equilibrium] \label{def-gen-nash-context} Suppose each player's move is a good move in the unilateral context, that is,
\[ 
x_i \in \varepsilon_i (\mathcal U^q_i (\mathbf x))
\]
for all players $1 \leq i \leq n$, where $x_i$ is the $i$th component of the tuple $\mathbf x$.
A strategy profile satisfying this condition for each player is called a selection equilibrium. 
\end{definition}

Our goal in this section is to show that selection equilibrium is a strict refinement of generalised Nash equilibrium. We start by showing that every selection equilibrium is also an equilibrium in the sense of Definition \ref{def-gen-nash}.

\begin{theorem} Every selection equilibrium is a generalised Nash equilibrium.
\end{theorem}
\begin{proof} Recall that by definition, for every context $p$ we have 
\[ x \in \varepsilon_i (p)\implies p(x) \in \overline{\varepsilon_i} (p) \]
since $\overline{\varepsilon_i} (p) = \{ p (x) \mid x \in \varepsilon_i (p) \}$. Assuming that $\mathbf x$ is a selection equilibrium we have
\[ x_i \in \varepsilon_i (\mathcal U^q_i (\mathbf x)) \]
Therefore
\[ \mathcal U^q_i (\mathbf x) (x_i) \in \overline{\varepsilon_i} (\mathcal U^q_i (\mathbf x)) \]
It remains to note that $\mathcal U^q_i (\mathbf x) ( x_i) = q (\mathbf x)$, because $\mathbf x [i \mapsto x_i] = \mathbf x$.
\end{proof}

Let us again illustrate this new concept using our running example.

\begin{example}[Extended Battle of the Sexes IV] As before, to illustrate concepts, consider the wife and the strategy profile $\mathbf x = (B,C)$. Her selection function yields: 
\[ \varepsilon_w (\mathcal U^q_w (B,C)) 
   = \{ r_w \; | \; \mathcal U^q_w (B,C)(r_w) \in \max_{x \in X_w} \pi_1(\mathcal U^q_w (B,C)(x)) \}
   = \{C\}. \]
As it holds that her current choice $x_w = B$ does not belong to this set of preferred \emph{strategies}, the strategy profile would not be in selection equilibrium from the point of view of the wife. Looking from the point of view of the husband, on the strategy profile $\mathbf x = (B,C)$ we can calculate his preferred strategies using his selection function as
\[ \varepsilon_h (\mathcal U^q_h (B,C)) 
   = \{ r_h \; | \; \mathcal U^q_h (B,C)(r_h) \in \max_{x \in X_h} \pi_2(\mathcal U^q_h (B,C)(x)) \}
   = \{B,F\}. \]
But his current choice is $x_h  = C$, so he would also be tempted to change his mind and go to either the ballet with his wife, or to the football on his own.
\end{example}

In our version of the Battle of the Sexes game the set of generalised Nash equilibria and the set selection equilibria are in fact identical. This is not a coincidence, as the following theorem shows that for games based on the maximisation selection function the classical notion, and our two generalised notions coincide.


\begin{theorem} \label{thm-classical} In a strategic game (see Remark \ref{normal-form-remark}) the standard definition of Nash equilibrium and the equilibrium notions of Definitions \ref{def-gen-nash} and \ref{def-gen-nash-context} are equivalent. 
\end{theorem}
\begin{proof}
Suppose the set of outcomes $R$ is $\RR^n$ and that the selection functions $\varepsilon_i$ are $i$-$\argmax$, i.e. maximising with respect to $i$-th coordinate. Unfolding Definition \ref{def-gen-nash-context} and that of a unilateral context $\mathcal U^q_i (\mathbf x)$, we see that a tuple $\mathbf x$ is an equilibrium strategy profile if for all $1 \leq i \leq n$
\[ x_i \in i\textup{-}\argmax_{x \in X_i} q (\mathbf x [i \mapsto x])). \]
But $x_i$ is a point on which the function $p(x) = q (\mathbf x [i \mapsto x])$ attains its maximum precisely when $p(x_i) \in \max_{x \in X_i} p(x)$. Hence 
\[ q(\mathbf x) = q (\mathbf x [i \mapsto x_i]) = p(x_i) = \max_{x \in X_i} p(x) = \max_{x \in X_i} q (\mathbf x [i \mapsto x]) \]
which is the standard definition of a Nash equilibrium: for each player $i$, the outcome obtained by not changing the strategy, i.e. $q(\mathbf x)$, is the best possible amongst the outcomes when any other available strategy is considered, i.e. $\max_{x \in X_i} q (\mathbf x [i \mapsto x])$. The same holds for our Definition \ref{def-gen-nash} based on the quantifier $\overline{\varepsilon_i}$.
\end{proof}

Theorem \ref{thm-classical} above shows that in the case of ``classical" strategic games the usual concept of a Nash equilibrium coincides with both the general Nash equilibrium and the selection equilibrium. On the other hand, for general games, Theorem 4.9 proves that every selection equilibrium is a generalised Nash equilibrium
\[
\text{selection equilibria} \subsetneq\text{generalised Nash equilibria}
\]
In Section \ref{sec:examples} we give several examples showing that the inclusion above is strict, i.e. that there are games where selection equilibrium is a strict refinement of generalised Nash equilibrium. 

We close this section with a last consideration of the classic Battle of the Sexes (Example \ref{ex-cbos}). Here, we analyze the game in a different way. We do not use utility functions and we do not use max operators as selection functions. 

\subsection{Battle of the Sexes -- Qualitatively}

Let us represent the Battle of the Sexes game in our framework in a truly idiomatic way, by removing numerical utilities completely and focussing only on the qualitative information. The choices of moves are still $X_w = X_h = \{ B, F \}$, but now the outcomes are merely a description of what happens, namely who goes to which event. We set $R = X_w \times X_h = \{ B, F \}^2$, so an element of $R$ is a pair where the first coordinate tells what the wife chose, and the second tells what the husband chose. Now the outcome function $q : X_w \times X_h \to R$ is simply the identity function. 

We build selection functions for each player in a compositional way, by observing that each player has a lexicographic preference: their first priority is to be coordinated, and all else being equal, their second priority is to go to their favourite event (ballet and football, respectively). We describe an element of $R$ as `coordinated' if its first coordinate equals its second coordinate, so the coordinated outcomes are $(B,B)$ and $(F,F)$. There is a selection function $\varepsilon_c$ that chooses all moves that lead to a coordinated outcome, which we can write as an inverse image
\[ \varepsilon_c (p) = p^{-1} ( \{ (B,B), (F,F) \}) \]
Next we have a pair of selection functions $\varepsilon_b, \varepsilon_f$ representing the `purely selfish' aims of attending ballet and football respectively:
\[ \varepsilon_b (p) = p^{-1} ( \{ (B,B), (B,F) \}) \]
\[ \varepsilon_f (p) = p^{-1} ( \{ (B,F), (F,F) \}) \]
Now we are ready to build our players' selection functions compositionally. Given a context, the joint selection function checks whether there are any moves which satisfy both `personalities' given by the coordinating and selfish selection functions. If so, the joint selection function returns those moves. If there are no moves satisfying both then the coordination takes priority, and the selfish aspect is ignored. Therefore the wife's selection function is
\[ \varepsilon_w (p) = \begin{cases}
\varepsilon_c (p) \cap \varepsilon_b (p) & \text{ if nonempty } \\
\varepsilon_c (p) & \text{ otherwise }
\end{cases} \]
and the husband's selection function is
\[ \varepsilon_h (p) = \begin{cases}
\varepsilon_c (p) \cap \varepsilon_f (p) & \text{ if nonempty } \\
\varepsilon_c (p) & \text{ otherwise }
\end{cases} \]
Observe that in building these selection functions we have \emph{not} assumed that the game's outcome function is the identity function, and so these selection functions will still work correctly if we change the rules of the game, for example if the couple have an agreement that sometimes forces them to go against their choices.

We will verify that $(B,B)$ is a selection equilibrium, and $(B,F)$ is not. For the strategy $(B,B)$ the wife's unilateral context is
\[ \mathcal U^{id}_w (B,B) (x) = (x,B) \]
The individual selection functions give
\[ \varepsilon_c (\mathcal U^{id}_w (B,B)) = \{ B \} \qquad \varepsilon_b (\mathcal U^{id}_w (B,B)) = \{ B \} \]
and so the wife's selection function gives
\[ \varepsilon_w (\mathcal U^{id}_w (B,B)) = \{ B \} \]
which means the wife has no incentive to deviate. The husband's unilateral context is
\[ \mathcal U^{id}_h (B,B) (x) = (B,x) \]
The individual selection functions give
\[ \varepsilon_c (\mathcal U^{id}_h (B,B)) = \{ B \} \qquad \varepsilon_f (\mathcal U^{id}_h (B,B)) = \{ F \} \]
Now we have a clash between the two personalities because these sets do not intersect. The coordination takes priority, and so
\[ \varepsilon_h (\mathcal U^{id}_h (B,B)) = \{ B \} \]
and the husband has no incentive to deviate. Therefore $(B,B)$ is a selection equilibrium.

For the strategy $(B,F)$ we see that the husband has an incentive to unilaterally deviate to $B$. His unilateral context is
\[ \mathcal U^{id}_h (B,F) (x) = (B,x) \]
The individual selection functions give
\[ \varepsilon_c (\mathcal U^{id}_h (B,F)) = \{ B \} \qquad \varepsilon_f (\mathcal U^{id}_h (B,F)) = \{ F \} \]
and so
\[ \varepsilon_h (\mathcal U^{id}_h (B,F)) = \{ B \} \]
Therefore the husband has an incentive to deviate to $B$, and so $(B,F)$ is not a selection equilibrium.

Although this is a trivial example, we believe that this method of modelling will distinguish itself from utility-based methods in its ability to scale easily to very complex situations. A realistic agent may have many competing aims, some context-dependent and some context-independent (for example immediate profit, long-term profit, fairness concerns, environmental concerns). Using selection functions allows us to treat each aim individually, and then afterwards combine them (with rules for breaking ties, such as the lexicographic rule in this example) into a realistic description of the agent. We will discuss these ideas in the conclusion.


\section{Examples}
\label{sec:examples}

In this section we continue to work on Example \ref{keynesian-ex1} where three judges $J=\{J_1, J_2, J_3\}$ vote for two contestants $A$ and $B$. The set of possible outcomes is $X=\{A, B\}$ and the actual outcome is determined by the majority function $\maj : X \times X \times X \rightarrow X$. We have chosen to work with three judges and two possible outcomes in order to simplify the exposition, so that $\maj$ becomes a total function. 

In contrast to Example \ref{keynesian-ex1}, where the decisions of judges 2 and 3 were fixed, in the current section we analyse a  game where all three players make strategic decisions. We analyse several instances of this game with different motivations of players 
in order to illustrate the expressiveness of selection functions. Moreover, we explore the distinction between Nash equilibria and selection equilibria. Our specific examples contain implausible Nash equilibria which are not selection equilibria. Indeed, we have examples in which every strategy is a Nash equilibrium, but there are few selection equilibria.

\subsection{Games with Context-independent Selection Functions}

In a strategic game, the judges rank the contestants according to a preference ordering. For example, suppose judges 1 and 2 prefer $A$ and judge 3 prefers $B$. Thus for each judge we have an order relation on $X$. Suppose the order relation of the first judge is $B \preceq_1 A$, the second judge is $B \preceq_2 A$ and the third is $A \preceq_3 B$.

The judges now attempt to maximise the outcome with respect to their preferred ordering. Hence we obtain 3 different selection functions, which are maximisation with respect to each ordering:
\begin{align*}
\varepsilon_1 (p)&= \argmax_{x_1 \in (X, \preceq_1)} p(x_1) \\
\varepsilon_2 (p)&= \argmax_{x_2 \in (X, \preceq_2)} p(x_2) \\
\varepsilon_3 (p)&= \argmax_{x_3 \in (X, \preceq_3)} p(x_3).
\end{align*}
In this particular example (but not in general) we can fix a `global' order $B \preceq A$ and notice that $\preceq_3$ is the dual order. Thus we can for short refer to $\varepsilon_1$ and $\varepsilon_2$ as $\argmax$ and $\varepsilon_3$ as $\argmin$.

The game is represented in Table \ref{tab:classic-two}. Notice that Nash and selection equilibria coincide for this game, because it is a strategic game.

There is a subtle difference between this setup and the usual strategic game. In the classical approach, each judge's ordering would be seen as a preference relation. This would typically be used to derive payoffs \cite{OsbRub94}, which amounts to an order embedding of $X$ into $\mathbb R$. Here there are no payoffs: we directly maximise over the discrete order $X$.

\begin{table}[t!]
\caption{Agents: max, max, min}
\begin{center}
\begin{tabular}{|c|c||c|c||c|c|}\hline
Strategy 	& Outcome 	& Nash 	 	    & Defects       & Selection	    & Defects 	    \\ \hline \hline
$AAA$ 	    & $A$      	&  \checkmark	&		        & \checkmark	& 		        \\ \hline
$AAB$      	& $A$      	&  \checkmark 	&		        & \checkmark	& 		        \\ \hline
$ABA$ 		& $A$ 		&  -	        & $J_3$	        & -			    & $J_3$	        \\ \hline
$ABB$ 		& $B$		&  -       		& $J_2$	        & -			    & $J_2$	        \\ \hline
$BAA$ 		& $A$ 		&  -          	& $J_3$	        & -			    & $J_3$ 	    \\ \hline
$BAB$ 		& $B$ 		&  -    		& $J_1$	        & -			    & $J_1$ 	    \\ \hline
$BBA$ 		& $B$ 		&  -      		& $J_1$, $J_2$	& -			    & $J_1$, $J_2$ 	\\ \hline
$BBB$ 		& $B$ 		&  \checkmark	&		        & \checkmark    & 	        	\\ \hline
\end{tabular}
\end{center}
\label{tab:classic-two}
\end{table}

We now want to give the calculations of the Nash equilibria of Table \ref{tab:classic-two}
in the notation of selection functions and unilateral contexts. First we take a look at the Nash equilibrium $BBB$ with outcome $\maj(BBB)=B$ and give the rationale of player 1. The unilateral context of player 1 is
$$\mathcal U^{\maj}_1 (BBB)(x) = \maj(xBB) = B$$
meaning that in the given context the outcome is $B$ no matter what player 1 chooses to play. The minimisation quantifier applied to such unilateral context gives
$$\overline{\varepsilon_1}(\mathcal U^{\maj}_1 (BBB)) = \max_{\preceq_1}(\mathcal U^{\maj}_1 (BBB))=\{B\}$$
meaning that, in the given context, player 1's preferred outcome is $B$. Hence, we can conclude by $B = \maj(BBB)\in\overline{\varepsilon_1}(\mathcal U^{\maj}_1 (BBB)(x))=\{B\}$ 
that $B$ is a Nash equilibrium strategy for player 1. This condition holds for each player and allows us to conclude that $BBB$ is a Nash equilibrium. In a similar way way we see in
$$B = \maj(BBA)\notin\overline{\varepsilon_1}(\mathcal U^{\maj}_1 (BBA)(x)) = \{A\}$$
since $\mathcal U^{\maj}_1 (BBA)(x) = \maj(xBA) = x$, so that $BBA$ is not a Nash equilibrium. In other words, in the context $BBA$ player 1 has an incentive to change his strategy to $A$, so that the new outcome $\maj(ABA) = A$ is better than the previous outcome $B$.

\subsection{Keynesian Beauty Contest}\label{subsec:Keynes}

Let us now consider the Keynesian beauty contest as the paradigmatic example of a game with some players having context-dependent selection functions.
The first judge $J_1$ ranks the candidates according to a preference ordering $B \preceq A$. The second and third judges, however, are `Keynesian agents': they have no preference relations over the candidates per se, 
but want to vote for the winning candidate. 

The selection equilibria are precisely those in which $J_2$ and $J_3$ are coordinated, and $J_1$ is not pivotal in any of these. In the next section we will explain in more detail how the fixpoint selection function models coordination.
Table \ref{tab:fixfix} contains a summary of the equilibria.

Consider the strategy $AAA$, which is a selection equilibrium of this game. Suppose the moves of $J_1$ and $J_2$ are fixed, but $J_3$ may unilaterally change strategy. The unilateral context is
\[ \mathcal U^{\maj}_3 (AAA) (x) = \operatorname{maj} (AAx) = A \]
Thus the unilateral context is a constant function, and its set of fixpoints is
\[ \fix (\mathcal U^{\maj}_3 (AAA)) = \{ A \} \]
This tells us that $J_3$ has no incentive to unilaterally change to the strategy $B$, because he will no longer be voting for the winner.

On the other hand, for the strategy $ABB$ the two Keynes agents are indifferent, because if either of them 
unilaterally changes to $A$ then $A$ will become the majority and they will still be voting for the winner. This is still a selection equilibrium (as we would expect) because the unilateral context is the identity function, and in particular $B$ is a fixpoint.

There are two selection equilibria, $BAA$ and $BBB$, which are implausible in the sense that $J_1$ is not voting for his preferred candidate.  

\begin{table}[t!]
\caption{Agents: max, fix, fix}
\begin{center}
\begin{tabular}{|c|c||c|c||c|c|}\hline
Strategy 	& Outcome 	& Nash 	 	    & Defects	& Selection	        & Defects 		\\ \hline \hline
$AAA$ 	    & $A$      	&  \checkmark   &		    & \checkmark		& 				\\ \hline
$AAB$      	& $A$      	&  \checkmark   &		    & -					& $J_3$			\\ \hline
$ABA$ 		& $A$ 		&  \checkmark	&		    & -					& $J_2$			\\ \hline
$ABB$ 		& $B$ 		&  \checkmark   &		    & \checkmark		& 				\\ \hline
$BAA$ 		& $A$ 		&  \checkmark   &		    & \checkmark		& 	 			\\ \hline
$BAB$ 		& $B$ 		&  -  		    & $J_1$	    & -					& $J_1$, $J_2$ 	\\ \hline
$BBA$ 		& $B$ 		&  - 		    & $J_1$	    & -					& $J_1$, $J_3$ 	\\ \hline
$BBB$ 		& $B$ 		&  \checkmark 	&		    & \checkmark		& 				\\ \hline
\end{tabular}
\end{center}
\label{tab:fixfix}
\end{table}

We now calculate the Nash and the selection rationale for the strategy profile $AAB$ of the Keynesian player 3.
The outcome of $AAB$ is $\maj(AAB)=A$. The unilateral context of player 3 is 
\[ \mathcal U^{\maj}_3 (AAB)(x)=\operatorname{maj}(AAx)=A \]
meaning that the outcome is (still) $A$ if player 3 unilaterally changes from $B$ to $A$.
The minimisation quantifier applied to this context gives 
\[ \overline{\varepsilon_3}(\mathcal U^{\maj}_3 (AAB)) = \fix(\mathcal U^{\maj}_3 (AAB)) = \{A\} \]
meaning that $A$ is the outcome resulting from an optimal choice.
Hence, we can conclude by 
\[ A = \maj(AAB) \in \overline{\varepsilon_3}(\mathcal U^{\maj}_3 (AAB)) = \{A\} \]
that $B$ is a Nash equilibrium strategy for player 3.

The rationale for the selection equilibrium is as follows: the strategy 
$B \notin \varepsilon_3 (\mathcal U^{\maj}_3 (AAB)) = \fix(\mathcal U^{\maj}_3 (AAB)) = \{A\}$
meaning that $AAB$ is not a selection equilibrium.

\subsection{Coordination Game}

\begin{table}[t!]
\caption{Agents: fix, fix, fix}
\begin{center}
\begin{tabular}{|c|c||c|c||c|c|}\hline
Strategy 	& Outcome 	& Nash           	 	& Defects	& Selection     & Defects 	\\ \hline \hline
$AAA$ 	    & $A$  	    &  \checkmark    		&		    & \checkmark	& 	    	\\ \hline
$AAB$      	& $A$       &  \checkmark     		&		    & -				& $J_3$	    \\ \hline
$ABA$ 		& $A$ 		&  \checkmark      	    &		    & -				& $J_2$	    \\ \hline
$ABB$ 		& $B$ 		&  \checkmark       	&		    & -				& $J_1$	    \\ \hline
$BAA$ 		& $A$ 		&  \checkmark         	&		    & -				& $J_1$ 	\\ \hline
$BAB$ 		& $B$ 		&  \checkmark         	&		    & -				& $J_2$	    \\ \hline
$BBA$ 		& $B$ 		&  \checkmark         	&		    & -				& $J_3$ 	\\ \hline
$BBB$ 		& $B$ 		&  \checkmark         	&		    & \checkmark	& 	    	\\ \hline
\end{tabular}
\end{center}
\label{tab:coordination}
\end{table}

We consider a game where all agents act according to a fixpoint goal.
Judges $J_1, J_2$ and $J_3$ want to vote for the winner, so the selection functions are given by the fixpoint operator $(X \to X) \to \power{X}$.
%
%
As can be seen in Table \ref{tab:coordination}, the selection equilibria are exactly the coordinated strategies. Note that the fixpoint selection function models coordination, and the game in which all selection functions are fixpoints is a coordination game. This gives us a new perspective on the Keynesian beauty contest as a one-sided coordination game: the Keynesian agent would like to coordinate with the group, whereas the agents of the group are not interested in coordination.

This game is a good example of why ordinary Nash equilibria are not suitable for modelling games with context-dependent quantifiers: it can be seen in the table that every strategy is a Nash equilibrium of this game, but the selection equilibrium captures our intuition perfectly that the equilibria should be the strategy profiles that are maximally coordinated, namely $AAA$ and $BBB$.

\subsection{Anti-coordination Game}

\begin{table}[t!]
\caption{Agents: non-fix, non-fix, non-fix}
\begin{center}
\begin{tabular}{|c|c||c|c||c|c|}\hline
Strategy 	& Winner 	& Nash           	 	& Defects	& Selection      	& Defects       	\\ \hline \hline
$AAA$ 	    & $A$  	    &  \checkmark    		&		    & -					& $J_1$, $J_2$, $J_3$ \\ \hline
$AAB$      	& $A$      	&  \checkmark    		&		    & \checkmark		& 					\\ \hline
$ABA$ 		& $A$ 		&  \checkmark	       	&		    & \checkmark		& 					\\ \hline
$ABB$ 		& $B$ 		&  \checkmark      		&		    & \checkmark		& 					\\ \hline
$BAA$ 		& $A$ 		&  \checkmark        	&		    & \checkmark		& 					\\ \hline
$BAB$ 		& $B$ 		&  \checkmark           &		    & \checkmark		& 					\\ \hline
$BBA$ 		& $B$ 		&  \checkmark        	&		    & \checkmark		&  					\\ \hline
$BBB$ 		& $B$ 		&  \checkmark        	&		    & -					& $J_1$, $J_2$, $J_3$ \\ \hline
\end{tabular}
\end{center}
\label{tab:anticoordination}
\end{table}

Just as the fixpoint selection function models coordination, so there is a `non-fixpoint' selection function which models \emph{anti-coordination} (or \emph{differentiation} as in the minority 
game \cite{arthur_inductive_1994,kets_minority_2007}). The set of non-fixpoints of a function $p : X \to X$ is
\[ \nonfix(p)= \{ x \in X \mid x \neq p(x) \} \]
Obviously this set might be empty, for instance when $p$ is the identity function. Hence, we extend this to a total selection function by specifying that the player is indifferent in the event that there are no non-fixpoints
\[ \varepsilon(p)= \begin{cases}
\nonfix(p) &\text{ if } \nonfix(p) \neq \varnothing \\
X &\text{ otherwise}
\end{cases} \]
Unlike for fixpoints, this selection function does not attain itself when considered as a quantifier. 
The corresponding quantifier is instead:
\[ \overline{\varepsilon}(p) = \{ p(x) \mid x \neq p(x) \} \]
In a game such as an election, an agent whose selection function is non-fix is a `punk' who aims to be in a minority. Therefore, let us consider the game in which all three judges are punks (see Table \ref{tab:anticoordination}). Of course only one can actually be in a minority, so the selection equilibria are precisely the `maximally anti-coordinated' strategy profiles, namely those in which one judge differs from the other two. This is another example of a game in which every strategy is a Nash equilibrium, but the selection equilibrium corresponds perfectly to our intuition.

\section{Compiling Games}
\label{sec:compiling}

In Section \ref{subsec:Keynes}, we modeled the Keynesian beauty contest. We showed that the set of Nash equilibria and the set of selection equilibria do not coincide. In Section \ref{sec:context_Nash} we proved that both equilibrium concepts are isomorphic when we only consider max and argmax as quantifier, respectively selection function. The natural question for the Keynesian beauty contest is: Is there a way to model the game using standard payoff functions and what is the set of equilibria we get?

For the reasons exposed in Section \ref{sec:context_dec}, there is no general utility representation of the fixpoint goals. As long as the outcomes are defined in terms of the candidate who wins the contest, it is not possible to consider an ordering of these outcomes alone but contextual information is also important.  

There are two things an analyst can do. First, he can redefine the outcome space and define new utility functions on these outcomes for the Keynes agents. This, however, would make it necessary to change the outcome function globally for all players. Secondly, for the given outcome function, he can find payoffs by hand that mimick an agent's fixpoint goals. That would mean to give up a compact and meaningful description of an agent via a utility function. Moreover, for each new game payoffs have to be calculated again.

\begin{table}[t!]
\caption{Agents: max, fix, fix}
\begin{center}
\begin{tabular}{|c|c||c|c||c|c|c|}\hline
Strategy 	& Outcome   & NE sim. S NE    & Defects   & $P_{J_1}$ & $P_{J_2}$	& $P_{J_3}$ 	\\ \hline \hline
$AAA$ 	    & $A$       & \checkmark     &           & 1			& 1			& 1				\\ \hline
$AAB$      	& $A$       & -     & $J_3$     & 1			& 1			& 0				\\ \hline
$ABA$ 		& $A$       & -     & $J_2$     & 1			& 0			& 1				\\ \hline
$ABB$		& $B$       & \checkmark     &           & 0			& 1			& 1				\\ \hline
$BAA$ 		& $A$       & \checkmark     &           & 1			& 1			& 1				\\ \hline
$BAB$ 		& $B$       & -     & $J_1,J_2$ & 0			& 0			& 1				\\ \hline
$BBA$ 		& $B$       & -     & $J_1,J_3$ & 0			& 1			& 0				\\ \hline
$BBB$ 		& $B$       & \checkmark    &           & 0			& 1			& 1				\\ \hline
\end{tabular}
\end{center}
\label{tab:compile}
\end{table}

Table \ref{tab:compile} illustrates this approach for the Keynesian beauty contest. The last three columns depict the computed payoff matrix. In column `NE sim. SNE' (`Nash equilibrium simulating selection equilibrium') we denote the Nash equilibria. Comparing it to Table \ref{tab:fixfix} shows that the Nash equilibria now simulate the selection equilibria. To attain the payoff functions for jugdes 2 and 3 (judge 1 prefers candidate $A$ over $B$), the analyst solves the game for their fixpoint goals. 

Consider the strategy profiles $AAA$ and $AAB$. These strategy profiles yield the same outcome: $A$ is the winner of the contest. In this situation judge 3 is not pivotal; independently of his action, $A$ will win the contest. Now, to achieve the context-dependency in this  game using payoff functions, the analyst has to introduce differences in the payoffs for the same final outcome. I.e., he provides different payoffs for the different actions leading to the same outcome and is thus implicitly using contextual information.

Contrast this with the selection function approach. Selection functions can not only be defined on the set of outcomes but on the whole function space of unilateral contexts. Thus, when considering the beauty contest, the outcome function is the same for an agent who is merely concerned with the final outcomes, such as judge 1 in our example, or who is concerned with the context, like the fixpoint judges 2 and 3.
 Selection functions do not convey implicitly the context-dependency but explicitly highlight it.

Fixpoint goals are but one example. In general, the feature that we can equip players with goals depending on the game description itself (described by the outcome function) makes the selection approach more expressive and more compositional because there is no need to change the global outcome function but only the local representation of one player.\footnote{Or, alternatively, there is no need to solve the game and encode context-dependent information by hand.} As a further consequence, this widens the possibilites to do comparative statics within the same game.

In the case an analyst wants to encode context-dependent information in a given game, there is a very simple procedure for calculating the payoffs such that the Nash equilibrium mimicks the selection one: we write $0$ whenever the player defects and $1$ whenever they do not defect (cf Table \ref{tab:compile}).

Besides convenience, there is another reason this is important: it gives a very simple way to talk about mixed strategies. A good example of this is the beauty context variant $(\max,\min,\fix)$. In this game $J_1$ prefers $A$, $J_2$ prefers $B$ and $J_3$ would like to vote for the winner. Intuitively we expect this game to have an equilibrium in which $J_1$ votes for $A$, $J_2$ votes for $B$ and $J_3$ mixes with arbitrary probability. By `compiling' the game to a classical utility function representation we regain these mixed equilibria. Although we are working on directly representing mixed strategies in the selection function framework, using this procedure is simple and effective. In particular, if we begin with a finite game and compile it, we obtain a strategic game to which Nash's theorem applies (whereas there is no reason to expect our original game to have any equilibria).

\ignore{
Since this procedure can be applied to any game, in particular we could apply it to a game that is already classical. In this case the translation preserves qualitative information about the equilibria, but quantitative information is destroyed. For example, Pareto dominance is not preserved by the translation.
}

\begin{theorem}
Consider a game $\mathcal G$ defined by total selection functions $\varepsilon_i : (X_i \to R) \to \mathcal P X_i$ and outcome function $q : \prod_i X_i \to R$. We define a strategic game $\mathcal G'$ to have the same move sets $X_i$, and the utility for the $i$th player of the strategy $\mathbf x : \prod_i X_i$ is defined by the outcome function
\[ q (\mathbf x)_i = \begin{cases}
1 &\text{ if } \mathbf x_i \in \varepsilon_i (\mathcal U^q_i \mathbf x) \\
0 &\text{ otherwise}
\end{cases} \]
Then
\begin{enumerate}
\item A strategy $\mathbf x$ of $\mathcal G'$ is a Nash equilibrum iff all players receive utility 1
\item The Nash equilibria of $\mathcal G'$ are exactly the selection equilibria of $\mathcal G$
\end{enumerate}
\end{theorem}

\begin{proof}
Notice that (2) follows immediately from (1), since by the construction of $q'$ a strategy $\mathbf x$ is a selection equilibria of $\mathcal G$ iff all players receive utility 1 in $\mathcal G'$. To prove (1) we first note that if all players receive 1 utility from a strategy $\mathbf x$ then trivially it is a Nash equilibrium, since getting more than 1 utility is impossible. Now suppose we have a strategy $\mathbf x$ where some player (say the $i$th) receives utility 0, and we will prove that $\mathbf x$ it is not a Nash equilibrium of $\mathcal G'$. Let $p$ be the unilateral context $p = \mathcal U^q_i (\mathbf x)$. Since a selection function can never return the empty set, choose some $x' \in \varepsilon_i (p)$. In particular, $x' \neq \mathbf x_i$ since otherwise we would have $q' (\mathbf x)_i = 1$. Let $\mathbf x' = \mathbf x [i \mapsto x']$ be the strategy in which player $i$ unilaterally changed from $\mathbf x_i$ to $x'$, and let $p' = \mathcal U^q_i (\mathbf x')$ be its unilateral context. Since $\mathcal U^q_i (\mathbf x)$ is independent of the value of $\mathbf x_i$, we have an equality of contexts $p = p'$, and so in particular $\varepsilon_i (p) = \varepsilon_i (p')$. Therefore $x' \in \varepsilon_i (p')$, so $q' (\mathbf x')_i = 1$, which proves that player $i$ can increase his utility from 0 to 1 by unilaterally changing from $\mathbf x$ to $x'$.
\end{proof}

\section{Conclusions}
\label{sec:conclusion}

In this paper we propose the use of quantifiers and selection functions to model individual choices and games. Our framework instantiates preferences and utility functions as a special case but we can also model agents whose preferences are incomplete, who deviate from maximization,  whose motives are not only influenced by the outcome but also by the way outcomes realize.  

In practice, economists often restrict themselves to model decisions or interactions using utility functions. We think this practice has two main negative side-effects our framework helps to overcome.

First, sticking to utility functions excludes interesting phenomena from analysis. The descriptions and explanations of behavior exclusively live in the framework of full optimization (or equivalently in having rational preferences). Yet, empirical evidence suggests that people's behavior deviates from this benchmark. Our framework contributes by introducing a whole new set of alternatives to describe behavior. 

Secondly, utility functions and preferences are but one formal encoding of economic situations. Even if a formal equivalence between utility maximization and selection functions exists, the naturalness of representing the economic problem via selection functions may be different and more insightful. Selection functions allow a high-level and a more abstract description than utility functions. We illustrated this focusing on one example: fixpoint selection functions as a high-level representation of coordination goals.  
The abstraction has a further crucial advantage: It introduces compositionality. Whereas with utility functions it can be necessary to either change the outcome space or change payoffs for a given outcome function by hand in order to model behavior, with selection functions the global outcome function remains unchanged and only the local selection functions have to be changed. 

More generally, our modelling perspective on the different levels of abstraction is a well known pattern studied extensively in computer science.
Any mathematical or logical language for knowledge representation \cite{harmelen_handbook_2007} faces a trade off between the goals of representation and reasoning \cite{brachman_knowledge_2004}, page 327:
"[$\hdots$]  why do we not attempt to define a formal knowledge representation language 
that is coextensive with a natural language like English? [$\hdots$] 
Although such a highly expressive language would certainly be desirable 
from a \emph{representation} standpoint, 
it leads to serious difficulties from a \emph{reasoning} standpoint."

The difference between using utility functions and selection functions 
is that both differ in their representational power. 
The natural language description of an economic situation 
can be more directly translated into the high level formal system of selection functions 
as opposed to the low level classical approach.
We have depicted both in Figure \ref{fig:model}.
In our approach, the payoff matrix as in Table \ref{tab:compile}, can be automatically computed, 
and the modeller needs only to decide which selection functions represent the agents and 
which outcome function represents the situation. 
The classical approach is depicted in Figure \ref{fig:model} by a translation into the payoff matrices.

Our approach of context-dependent modeling of the left part of Figure \ref{fig:model}
can thus be seen as a modelling technique to reduce the gap between the high level
description of an economic situation in natural language and the formal modelling language.
Selection functions are mid-level in between the high level natural language and the 
low level language of utility maximization represented traditionally in calculus in classical games.

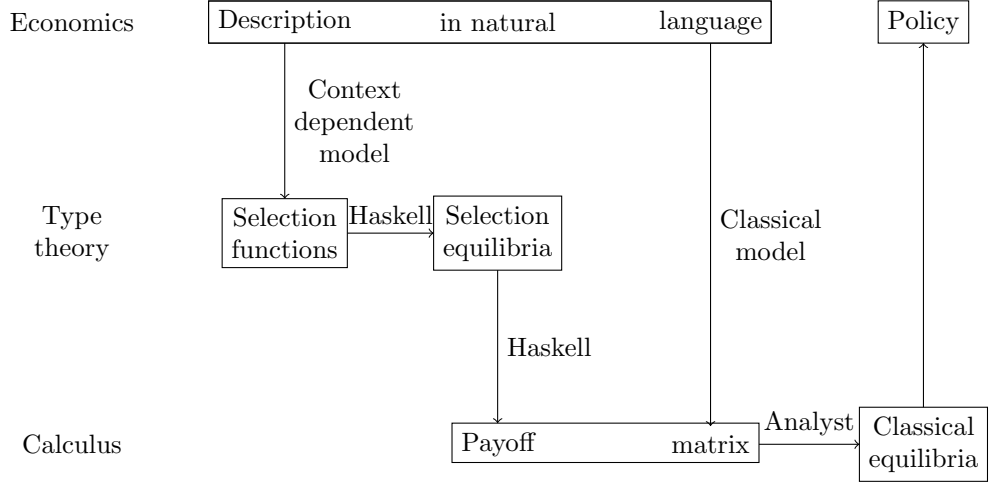
\begin{figure}[t!]
\usetikzlibrary{calc}
\begin{center}
\caption{Context depended and classical modelling}
\footnotesize
\begin{tikzpicture}[node distance=2.8cm, auto]
\node (A) {Economics};
\node (B) [below of=A] [align=center] {Type\\theory};

\node (C) [below of=B] {Calculus};

\node (D) [right of=A] {Description};
\node (E) [right of=D] {in natural};
\node (J) [right of=E] {language};
\draw ($(D.north west)$) rectangle ($(J.south east)$);

\node (F) [below of=D] [align=center] {Selection\\functions};
\draw ($(F.north west)$) rectangle ($(F.south east)$);

\node (G) [right of=F] [align=center] {Selection\\equilibria};
\draw ($(G.north west)$) rectangle ($(G.south east)$);

\node (H) [below of=G] {Payoff};
\node (I) [right of=H] {matrix};
\draw ($(H.north west)$) rectangle ($(I.south east)$);

\node (K) [right of=I] [align=center] {Classical\\equilibria};
\draw ($(K.north west)$) rectangle ($(K.south east)$);

\node (L) [right of=J] {Policy};
\draw ($(L.north west)$) rectangle ($(L.south east)$);

\draw ($(D.north west)$) rectangle ($(J.south east)$);

\draw [->] (J) to node [align=center] {Classical\\model} (I);
\draw [->] (D) to node [align=center] {Context\\dependent\\model} (F);
\draw [->] (F) to node {Haskell} (G);
\draw [->] (G) to node {Haskell} (H);
\draw [->] (I) to node {Analyst} (K);
\draw [->] (K) to node {} (L);
\end{tikzpicture}
\label{fig:model}
\normalsize
\end{center}
\end{figure}

An account for the representational power of languages is an involved research topic
and even more raising the expressivity while not sacrificing reasonability.
However, our hypothesis is that this higher-order approach increases the expressivity of the representation language for game theory while not sacrificing reasoning possibilities.

In fact, our framework is ready for automated reasoning, as opposed to calculus \cite{pavlovic_calculus_2002}. Our approach is directly programmable in modern functional languages like Haskell that has been developed within the high level modern mathematical type theory \cite{pierce2002types} in order to increase the expressivity of imperative languages such as Fortran or Matlab.  The essence of functional languages is the usage of higher order functions that take and output other functions like selection functions, quantifier and outcome functions. Functional programming languages can be understood as languages to design languages. A typical approach to programing via functional languages is to design a domain specific language that allows to express the problems in a most direct and natural way while the program is compiled with its problem declaration into the low level solution steps of an intermediate imperative or a very low level machine language. This is much in the spirit of our approach to compile the selection function model into the payoff matrix of a classical representation of the Keynesian beauty contest game in Table \ref{tab:compile}.

Regarding automation, we have heavily taken advantage of a prototype Haskell tool that calculates equilibria by brute force (enumeration of all strategies) for the games we have discussed in this paper. In fact the discovery of the notion of selection equilibrium 
has been a direct consequence of using the tool.
Before using the software we were misguided by our intuition,
and did not recognise the difference between Nash equilibria and selection equilibria.

There are several avenues for future research. In this paper, we introduced a framework, so clearly the overall usefulness of our approach will depend on providing interesting applications. Secondly, we will extend the theory of selection functions to other classes of games, such as repeated and sequential games, as well as games of incomplete information. 
Lastly, we see a huge potential in using functional languages to model games, in particular large and complex games. We are currently working on a prototpye Haskell tool that allows us to do exactly that. Moreover we are working on the use of monads (a very powerful paradigm for programming with higher order functions, implemented with great success in Haskell) to model computational effects in economics such as non-determinism, stochastic choice, memory and global state. 

\bibliographystyle{plain}
\bibliography{references_2}



\end{document}